\documentclass[journal]{IEEEtran}

\usepackage[colorlinks]{hyperref}
\usepackage{amsmath}
\usepackage{amsfonts}
\usepackage{amssymb}
\usepackage{amsthm}
\usepackage{tabularx}
\usepackage{mathrsfs} 
\usepackage{algorithmic}
\usepackage{algorithm} 
\usepackage{mathtools}
\usepackage{url}
\usepackage{hyperref}

\newtheorem{Proposition}{Proposition}

\usepackage{stfloats}
\usepackage{float}
\usepackage{graphicx}
\usepackage{xcolor}
\usepackage{subfigure}
\renewcommand{\vec}[1]{\mathbf{#1}}

\makeatletter
\def\blfootnote{\xdef\@thefnmark{}\@footnotetext}
\makeatother

\begin{document}
\title{\LARGE On Performance of RIS-Aided Fluid Antenna Systems}

\author{Farshad~Rostami~Ghadi,~\IEEEmembership{Member},~\textit{IEEE}, 
            Kai-Kit~Wong, \IEEEmembership{Fellow}, \textit{IEEE}, 
	    Wee~Kiat~New,~\IEEEmembership{Member},~\textit{IEEE}, 
	    Hao~Xu,~\IEEEmembership{Member},~\textit{IEEE}, 
	    Ross Murch,~\IEEEmembership{Fellow},~\textit{IEEE}, 
	    and Yangyang Zhang

\thanks{The work of F. Rostami Ghadi, W. K. New, H. Xu, and K. K. Wong is supported by the Engineering and Physical Sciences Research Council (EPSRC) under Grant EP/W026813/1.}
\thanks{F. Rostami Ghadi, W. K. New, H. Xu, and K. K. Wong are with the Department of Electronic and Electrical Engineering, University College London, London WC1E 6BT, U.K.}
\thanks{K.-K. Wong is also affiliated with Yonsei Frontier Laboratory, Yonsei University, Seoul, 03722, Korea. (e-mail: $\rm kai\text{-}kit.wong@ucl.ac.uk$).}
\thanks{R. Murch is with the Department of Electronic and Computer Engineering and Institute for Advanced Study (IAS), Hong Kong University of Science and Technology, Clear Water Bay, Hong Kong SAR, China.}
\thanks{Y. Zhang is with Kuang-Chi Science Limited, Hong Kong SAR, China.}

\thanks{Corresponding authors: F. Rostami Ghadi and Kai-Kit Wong.}

\vspace{-2 em}}

\maketitle

\begin{abstract}
This letter studies the performance of reconfigurable intelligent surface (RIS)-aided communications for a fluid antenna system (FAS) enabled receiver. Specifically, a fixed single-antenna base station (BS) transmits information through a RIS to a mobile user (MU) which is equipped with a planar fluid antenna in the absence of a direct link. We first analyze the spatial correlation structures among the positions (or ports) in the planar FAS, and then derive the joint distribution of the equivalent channel gain at the user by exploiting the central limit theorem. Furthermore, we obtain compact analytical expressions for the outage probability (OP) and delay outage rate (DOR). Numerical results illustrate that using FAS with only one activated port into the RIS-aided communication network can greatly enhance the performance, when compared to traditional antenna systems (TAS).
\end{abstract}


%
\blfootnote{Digital Object Identifier 10.1109/XXX.2024.XXXXXXX}

\section{Introduction}\label{sec-intro}  
\IEEEPARstart{I}{n recent} years, the reconfigurable intelligent surface (RIS) has emerged as a promising technology to greatly extend the coverage region in future wireless communication systems \cite{basar2019wireless}. RISs are artificial surfaces with a large number of low-cost reflecting elements that can control the wireless signal propagation environment to redirect radio waves from base stations (BSs) to targeted mobile users (MUs). One of the key drivers of RIS is that perpetual operation is becoming more feasible due to recent advances in low-power electronics \cite{abadal2020programmable}. Since its inception, much has been researched for RIS and great progress has been made towards estimating the cascaded channel state information (CSI) for optimization \cite{Wei-2021,Liu-2020}.


On the other hand, the fluid antenna system (FAS) has arisen as an enhancement method for multiple-input multiple-output (MIMO) by introducing a new degree of freedom via antenna position flexibility \cite{New-twc2023}. FAS represents all forms of movable and non-movable position-flexible antenna systems \cite{Wong-2020cl,wong2020fluid} and was first introduced by Wong {\em et al.}~in \cite{Wong-2022fcn,Wong-PartI2023}. Recent results have studied the achievable performance of FAS and reported impressive gains \cite{Khammassi-2023,Vega-2023,Psomas-dec2023}. Some of the related results came under the name of `movable' antenna systems that could be viewed as a particular example of FAS \cite{Zhu-Wong-2024}.

While RIS and FAS are complementary technologies to each other, the synergy between them is not well understood. The only result when considering both RIS and FAS was reported in \cite{Wong-PartIII2023} where RISs were used as artificial scatterers to generate more multipath so that FAS could work effectively in differentiating different user signals for multiuser communications. In \cite{Wong-PartIII2023}, RISs were randomized and not optimized. Different from the previous work, in this letter, we consider a single-user RIS-aided communication channel in which a single fixed-position antenna BS sends a message to a FAS-equipped user via an optimized RIS. The direct link is broken and the signal has to come through the RIS. Also, the user has a one-sided, planar FAS capable of switching to the best position in a prescribed two-dimensional (2D) space for reception.

The technical contributions of this letter are as follows. After characterizing the spatial correlation between the positions (also known as `ports') of the FAS, we first derive analytical expressions for the distributions of equivalent channels at the user by utilizing the Gaussian copula. Then we obtain compact expressions for the outage probability (OP) and delay outage rate (DOR). Our numerical results indicate that deploying FAS at the MU can make the RIS much more effective when compared to using a traditional antenna system (TAS). 
 
\begin{figure}[!t]
\centering
\includegraphics[width=1\columnwidth]{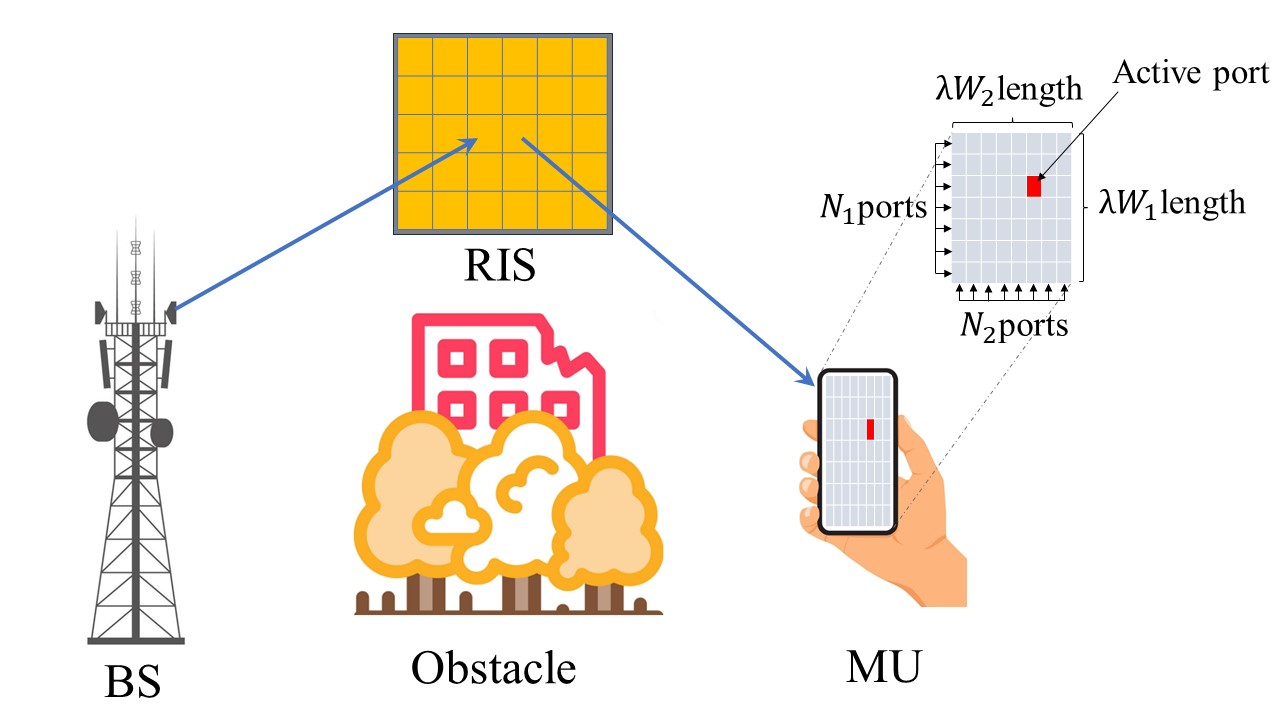}
\caption{A RIS-aided channel from a TAS BS to a FAS-equipped user.}\vspace{0cm}\label{fig-model}
 \end{figure}

\section{System Model}\label{sec-sys}
We consider a wireless communication scenario as shown in Fig. \ref{fig-model}, where a single-TAS BS sends a message $x$ with transmit power $P$ to a planar FAS-equipped MU with the help of a RIS that is composed of $M$ reflecting elements. Suppose that the RIS is there to establish the link between the BS and the user, as the direct link is blocked. We assume that perfect CSI is available at the RIS to optimally configure the phase shift matrix. Additionally, we consider a grid structure for the respective FAS so that $N_l$ ports are uniformly distributed along a linear space of length $W_l\lambda$ for $l\in\{1,2\}$, i.e., $N=N_1\times N_2$ and $W=W_1\lambda\times W_2\lambda$, where $\lambda$ denotes the wavelength associated with the carrier frequency. Moreover, an applicable mapping function as $\mathcal{F}\left(n\right)=\left(n_1,n_2\right)$ and $\mathcal{F}^{-1}\left(n_1,n_2\right)=n$ is supposed to convert the 2D indices to the one-dimensional (1D) index, in which $n\in\{1,\dots,N\}$ and $n_l\in\{1,\dots,N_l\}$. Therefore, the received signal at the $n$-th port of FAS of the user can be written as
\begin{align}
y_n=\underbrace{\tilde{\vec{h}}_{\mathrm{BS\text{-}RIS}}^T\vec{\Psi}\vec{h}_{\mathrm{RIS\text{-}MU},n}}_{\triangleq h_n}x+z_n,
\end{align}
 where $h_n$ represents the equivalent channel between the BS and the $n$-th FAS port at the user. More specifically, the vectors $\tilde{\vec{h}}_\mathrm{BS\text{-}RIS}=d_\mathrm{BS\text{-}RIS}^{-\alpha}\left[\tilde{h}_1\mathrm{e}^{-j\theta_1},\dots,\tilde{h}_M\mathrm{e}^{-j\theta_M}\right]^T\in\mathbb{C}^{M\times 1}$ and $\vec{h}_{\mathrm{RIS\text{-}MU},n}=d_\mathrm{RIS\text{-}MU}^{-\alpha}\left[h_{1,n}\mathrm{e}^{-j\beta_1},\dots,h_{M,n}\mathrm{e}^{-j\beta_M}\right]^T\in\mathbb{C}^{M\times 1}$ include the channel gains from the BS to the RIS (i.e., $\mathrm{BS\text{-}RIS}$) and from the RIS to the user (i.e., ($\mathrm{RIS\text{-}MU}$)), respectively. The terms $d_{\mathrm{BS\text{-}RIS}}$ and $d_{\mathrm{RIS\text{-}MU}}$ are the distances of $\mathrm{BS\text{-}RIS}$ and $\mathrm{RIS\text{-}MU}$, respectively, $\alpha>2$ denotes the path-loss exponent, $\theta_m$ and $\beta_m$ are the phases of the respective channel gains, and $\tilde{h}_m$ and $h_{m,n}$ define the amplitudes of the corresponding channel gains. Besides, the diagonal matrix $\vec{\Psi}=\text{diag}\left(\left[\zeta_1\mathrm{e}^{j\psi_1},\dots,\zeta_M\mathrm{e}^{j\psi_M}\right]\right)\in\mathbb{C}^{M\times M}$ contains the adjustable phases by the reflecting elements of the RIS in which $\zeta_m(\psi_m)=1$. Also, $z_n$ is the independent identically distributed (i.i.d.) additive white Gaussian noise (AWGN) with zero mean and variance $\sigma^2$ at each FAS port of the user. 
 
Given that the ports can be arbitrarily close to each other, the channel coefficients $h_n$ are spatially correlated. By considering the fact that a planar FAS can only have a $180^{\circ}$ reception coverage of the radio environment (i.e., \textit{half-space} in front),\footnote{A noticeable difference here is that in \cite{New-twc2023}, even a planar FAS is considered, the radio waves are assumed to come in all $360^{\circ}$ directions.} the spatial correlation between any two ports $n=\mathcal{F}^{-1}\left(n_1,n_2\right)$ and $\tilde{n}=\mathcal{F}^{-1}\left(\tilde{n}_1,\tilde{n}_2\right)$ can be characterized as 
\begin{multline}\label{eq-cov}
\varpi_{n,\tilde{n}}\coloneqq\mathrm{cov}\left\{n,\tilde{n}\right\}\\
={\rm sinc}\left(\frac{2}{\lambda}\sqrt{\left(\frac{|n_1-\tilde{n}_1|}{N_1-1}W_1\right)^2+\left(\frac{|n_2-\tilde{n}_2|}{N_2-1}W_2\right)^2}\right),
\end{multline}
 in which $\mathrm{sinc(t)=\frac{\sin(\pi t)}{\pi t}}$ is the sinc function. The details of the proof can be found in Appendix \ref{app1}. Hence, after applying the mapping function, the spatial correlation matrix $\vec{R}$ can be expressed as
\begin{align}
\vec{R}=\begin{bmatrix}
 	\varpi_{1,1} & \varpi_{1,2} &\dots& \varpi_{1,N}\\
 	\varpi_{2,1} & \varpi_{2,2} &\dots& \varpi_{2,N}\\ \vdots & \vdots & \ddots & \vdots\\
 	\varpi_{N,1} & \varpi_{N,2} &\dots& \varpi_{N,N}
\end{bmatrix}.
\end{align}
 
By assuming that only the optimal port that maximizes the received signal-to-noise ratio (SNR) at the MU is activated, the resulting SNR at the user can be found as
\begin{align}
	\gamma=\frac{P\left|\left[h\right]_{n^*}\right|^2}{\sigma^2}=\overline{\gamma}\left|\left[{h}_i\right]_{k_i^*}\right|^2,
\end{align} 
in which $\bar{\gamma}=\frac{P}{\sigma^2}$ is the average SNR and $n^*$ defines the index of the optimally selected port at the user, i.e., 
\begin{align}
n^*=\arg\underset{n}{\max}\left\{\left|\left[h\right]_{n}\right|^2\right\}.
\end{align}
The resulting channel gain at the user can be expressed as 
\begin{align}
h_{\mathrm{FAS}}^2=\max\left\{|h_{1}|^2,|h_{2}|^2,\dots,|h_{N}|^2\right\}.
\end{align}

\section{Performance Analysis}
Here, we first characterize the cumulative distribution function (CDF) and probability density function (PDF) of the equivalent channel at the MU by exploiting the central limit theorem (CLT) and copula theory. Then we derive the OP and DOR in compact analytical expressions. 

\subsection{Statistical Characterization}
In order to determine the statistical distribution of the SNR at the user, we first need to obtain the marginal distribution  of $h_\mathrm{FAS}^2$ in the presence of the RIS. To this end, the fading channel coefficient can be rewritten as
\begin{align}
\left|h_{n}\right|^2&=\frac{\left|\sum_{m=1}^M\tilde{h}_mh_{m,n}\mathrm{e}^{-j(\psi_m-\theta_m-\beta_m)}\right|^2}{d_\mathrm{BS\text{-}RIS}^\alpha d_\mathrm{RIS\text{-}MU}^\alpha}\\
&\hspace{-.5mm}\overset{(a)}{=}\frac{\left(\sum_{m=1}^M\tilde{h}_mh_{m,n}\right)^2}{d_\mathrm{BS\text{-}RIS}^\alpha d_\mathrm{RIS\text{-}MU}^\alpha}=\frac{A^2}{d_\mathrm{BS\text{-}RIS}^\alpha d_\mathrm{RIS\text{-}MU}^\alpha},
\end{align}
in which $(a)$ is obtained from the assumption of perfect CSI for RIS configuration, which enables ideal phase shifting, i.e., $\psi_m=\theta_m+\beta_m$. Then, by the CLT for a large number of reflecting elements, i.e., $M\gg1$, $A=\sum_{m=1}^M\eta_m\kappa_m$ can be accurately approximated as a Gaussian random variable with $\mu_A=\frac{M\pi}{4}$ and variance of $\sigma^2_A=M\left(1-\frac{\pi^2}{16}\right)$ \cite{basar2019wireless}. Thus, $A^2$ is a non-central chi-square random variable with one degree of freedom and has the marginal CDF and PDF as \cite{proakis2008digital}
\begin{align}
F_{A^2}(a) = 1-Q_{\frac{1}{2}}\left(\sqrt{\frac{\tau}{\sigma_A^2}},\sqrt{\frac{a}{\sigma_A^2}}\right),
\end{align}
and
\begin{align}
f_{A^2}(a) = \frac{1}{2\sigma^2_\mathrm{A}}\left(\frac{a}{\tau}\right)^{-\frac{1}{4}}\mathrm{exp}\left(-\frac{a+\tau}{2\sigma^2}\right)\mathcal{I}_{-\frac{1}{2}}\left(\frac{\sqrt{a\tau}}{\sigma_A^2}\right),
\end{align}
where $\tau=\mu_\mathrm{A}^2$ is the non-centrality parameter and $Q_l(v,w)$ denotes the Marcum $Q$-function with order $l$. Therefore, the CDF of $h^2_\mathrm{FAS}(r)$ can be derived as
\begin{align}
F_{h_{{\mathrm{FAS}}}^2}(r)&=\Pr\left(h_{{\mathrm{FAS}}}^2\leq r\right)\notag\\
&=\Pr\left(\max\left\{|h_{1}|^2,\dots,|h_{N}|^2\right\}\leq r\right)\notag\\\notag
&=\Pr\left(\max\left\{A^2,\dots, A^2\right\}\leq r\tilde{d}\right)\\
&=F_{A^2,\dots,A^2}\left(r\tilde{d},\dots,r\tilde{d}\right),\label{eq-cdf1}
\end{align}
in which $\tilde{d}= d_\mathrm{BS\text{-}RIS}^\alpha d_\mathrm{RIS\text{-}MU}^\alpha$. Furthermore, to obtain the joint multivariate distributions provided in \eqref{eq-cdf1}, we exploit Sklar's theorem, which can accurately connect the joint multivariate distributions of arbitrary random variables to their marginal distribution with the help of a copula function. As a consequence, by exploiting the analytical results in \cite{ghadi2023gaussian}, $F_{h_{{\mathrm{FAS}}}^2}(r)$ can be derived as
\begin{align}
F_{h_{{\mathrm{FAS}}}^2}(r)=C\left(F_{A^2}\left(r\tilde{d}\right),\dots,F_{A^2}\left(r\tilde{d}\right);\vartheta_C\right),\label{eq-gen-cdf}
\end{align}
where $C(u_1,\dots,u_N)$ denotes the $N$-dimensional copula on the unit hypercube $[0,1]^N$ with uniformly distributed random variables over $[0,1]$ and $\vartheta_C$ is the copula parameter that can control the degree of dependency between fluid antenna ports. Then, by applying the chain rule to \eqref{eq-gen-cdf}, the corresponding PDF, $f_{h_{{\mathrm{FAS}}}^2}(r)$, can be obtained as
\begin{align}\label{eq-gen-pdf}
f_{h^2_{\mathrm{FAS}}}(r)=\prod_{\hat{n}=1}^{N}f_{A_{\hat{n}}^2}\left(r\tilde{d}\right) c\left(F_{A^2}\left(r\tilde{d}\right),\dots,F_{A^2}\left(r\tilde{d}\right);\vartheta_C\right),
\end{align}
where $c(u_1,\dots,u_N)$ represents the copula density function. 

Note that \eqref{eq-gen-cdf} and \eqref{eq-gen-pdf} are mathematically applicable under any choice $C$ that includes the copula properties \cite[Remark 3]{rostami2024physical}. In this regard, it was understood that the Gaussian copula is an appropriate model that can accurately describe the spatial correlation between fluid antenna ports and make the mathematical analysis more tractable. Therefore, $F_{h_{{\mathrm{FAS}}}^2}(r)$ and $f_{h_{{\mathrm{FAS}}}^2}(r)$ can be, respectively, expressed as 
\begin{multline}\label{eq-cdf-g}
F_{h_{{\mathrm{FAS}}}^2}(r)\\
=\Phi_\vec{R}\left(\varphi^{-1}\left(F_{A^2}\left(r\tilde{d}\right)\right),\dots,\varphi^{-1}\left(F_{A^2}\left(r\tilde{d}\right)\right);\vartheta_G\right)
\end{multline}
and
\begin{multline}\label{eq-pdf-gaussian}
f_{h_\mathrm{FAS}^2}\left(r\right)=\prod_{\hat{n}=1}^Nf_{A^2_{\hat{n}}}\left(r\tilde{d}\right)\\
\times\frac{\exp\left(-\frac{1}{2}\left(\boldsymbol{\vec{\varphi}}^{-1}_{A^2}\right)^T\left(\vec{R}^{-1}-\vec{I}\right)\boldsymbol{\vec{\varphi}}^{-1}_{A^2}\right)}{\sqrt{{\rm det}\left(\vec{R}\right)}},
\end{multline}
in which $\mathrm{det}\left(\vec{R}\right)$ denotes the determinant of the correlation matrix $\vec{R}$, $\vec{I}$ is the identity matrix, and $\varphi^{-1}\left(F_{A^2}\left(r\tilde{d}\right)\right)=\sqrt{2}\mathrm{erf}^{-1}\left(2F_{A^2}\left(r\tilde{d}\right)-1\right)$ is the quantile function of the standard normal distribution, where $\mathrm{erf}^{-1}\left(\cdot\right)$ is the inverse of the error function $\mathrm{erf}\left(z\right)=\frac{2}{\sqrt{\pi}}\int_0^z\mathrm{e}^{-t^2}dt$. The term  $\Phi_\vec{R}(\cdot)$ is the joint CDF of the multivariate normal distribution with zero mean vector and correlation matrix $\vec{R}$,  $\vartheta_{G}$ denotes the correlation parameter of the Gaussian copula, and $\boldsymbol{\vec{\varphi}}^{-1}_{A^2}=\left[\varphi^{-1}\left(F_{A^2}\left(r\tilde{d}\right)\right),\dots,\varphi^{-1}\left(F_{A^2}\left(r\tilde{d}\right)\right)\right]^T$. Hence, $F_{h_{{\mathrm{FAS}}}^2}(r)$ and $f_{h_{{\mathrm{FAS}}}^2}(r)$ are rewritten as \eqref{eq-cdf} and \eqref{eq-pdf} (see top of the next page), where we also have \eqref{eq-phi}.

\begin{figure*}[t]
\normalsize
\begin{align}\label{eq-cdf}
F_{g_{\mathrm{FAS}}}(r)= \Phi_{\vec{R}}\left(\sqrt{2}\mathrm{erf}^{-1}\left(1-2Q_{\frac{1}{2}}\left(\sqrt{\frac{\tau}{\sigma_A^2}},\sqrt{\frac{a}{\sigma_A^2}}\right)\right),\dots,\sqrt{2}\mathrm{erf}^{-1}\left(1-2Q_{\frac{1}{2}}\left(\sqrt{\frac{\tau}{\sigma_A^2}},\sqrt{\frac{a}{\sigma_A^2}}\right)\right);\vartheta_{G}\right)
\end{align}
\hrulefill
\begin{align}\label{eq-pdf}
f_{h_\mathrm{FAS}^2}\left(r\right)=\left[\frac{1}{2\sigma^2_\mathrm{A}}\left(\frac{a}{\tau}\right)^{-\frac{1}{4}}\mathrm{exp}\left(-\frac{a+\tau}{2\sigma^2}\right)\mathcal{I}_{-\frac{1}{2}}\left(\frac{\sqrt{a\tau}}{\sigma_A^2}\right)\right]^N
\frac{\exp\left(-\frac{1}{2}\left(\boldsymbol{\vec{\varphi}}^{-1}_{A^2}\right)^T\left(\vec{R}^{-1}-\vec{I}\right)\boldsymbol{\vec{\varphi}}^{-1}_{A^2}\right)}{\sqrt{{\rm det}\left(\vec{R}\right)}}
\end{align}
\hrulefill
\begin{align}\label{eq-phi}
\boldsymbol{\vec{\varphi}}^{-1}_{A^2}=\left[\sqrt{2}\mathrm{erf}^{-1}\left(1-2Q_{\frac{1}{2}}\left(\sqrt{\frac{\tau}{\sigma_A^2}},\sqrt{\frac{a}{\sigma_A^2}}\right)\right),\dots,\sqrt{2}\mathrm{erf}^{-1}\left(1-2Q_{\frac{1}{2}}\left(\sqrt{\frac{\tau}{\sigma_A^2}},\sqrt{\frac{a}{\sigma_A^2}}\right)\right)\right]^T
\end{align}
\hrulefill
\end{figure*}


\subsection{OP Analysis}
The OP is defined as the probability that the received SNR is less than a given SNR threshold $\gamma_\mathrm{th}$, i.e.,
$P_\mathrm{out}=\Pr\left(\gamma\le\gamma_\mathrm{th}\right)$. Using the results obtained above, the OP for the considered system model is provided in the following proposition. 

\begin{Proposition}
The OP for the considered RIS-aided FAS is given by
\begin{align}
P_\mathrm{out}=F_{h_{{\mathrm{FAS}}}^2}\left(\frac{\gamma_\mathrm{th}\tilde{d}}{\overline{\gamma}}\right),
\end{align}
where $F_{h_{{\mathrm{FAS}}}^2}(r)$ is defined in \eqref{eq-cdf}.
\end{Proposition}

\begin{proof}
By definition of the OP and then applying the transformation $\gamma=\frac{A^2\overline{\gamma}}{\tilde{d}}$, the proof is completed. 
\end{proof}

\subsection{DOR Analysis}
The DOR is defined as the time it takes to successfully transmit a specific amount of data $R$ over a wireless channel with a bandwidth $B$, surpassing a predefined threshold duration $T_\mathrm{th}$, i.e., $\Pr\left(T_\mathrm{dt}>T_\mathrm{th}\right)$ in which
\begin{align}
T_\mathrm{dt}=\frac{R}{B\log_2\left(1+\gamma\right)}
\end{align}
indicates the delivery time \cite{ghadi2023gaussian}. The DOR for the considered system model is presented in the following proposition. 

\begin{Proposition}
The DOR for the considered RIS-aided FAS is given by
\begin{align}
P_\mathrm{dor}=F_{h^2_\mathrm{FAS}}\left(\frac{\tilde{d}\left(\mathrm{e}^{\frac{R\ln 2}{B T_\mathrm{th}}}-1\right)}{\overline{\gamma}}\right).
\end{align}
\end{Proposition}

\section{Numerical Results}
In this section, we present the numerical results to evaluate the performance of the proposed RIS-aided FAS in terms of the OP and DOR. To this end, we set the channel model parameters as $\sigma^2=120~{\rm dBm}$, $\alpha=2.5$, $d_\mathrm{BS\text{-}RIS}=d_\mathrm{RIS\text{-}MU}=2~{\rm km}$, $\gamma_\mathrm{th}=0~{\rm dB}$, $T_\mathrm{th}=3~{\rm ms}$, $R=3~{\rm kbits}$, $B=2~{\rm MHz}$, and $P=15~{\rm dBm}$. It is worth noting that the multivariate normal distribution in \eqref{eq-cdf} and the corresponding Gaussian copula are implemented numerically, exploiting the mathematical package and algorithm of the MATLAB programming language. Besides, TAS is considered as a benchmark for comparison.

Figs.~\ref{fig-o-p-n} and \ref{fig-o-p-w} illustrate the impact of transmit power $P$ on the OP behavior for different values of the fluid antenna port $N$ and fluid antenna size $W$, respectively. It is clearly seen that considering a larger number of ports, e.g., $N=4\times4$, or a higher value of fluid antenna size, e.g., $W=3\lambda\times3\lambda$, can significantly enhance the OP performance compared with the TAS. Moreover, we can observe that such an improvement becomes more noticeable when a larger number of RIS elements $M$ is considered. The key reason behind the behavior of the OP in Fig.~\ref{fig-o-p-n} is that while increasing $N$ with a constant $W$ enhances the spatial correlation among fluid antenna ports, it also has the potential to boost channel capacity, diversity gain, and spatial multiplexing simultaneously. Consequently, this can alleviate fading effects and improve the overall quality of the links. Furthermore, the primary factor influencing the behavior observed in Fig.~\ref{fig-o-p-w} is that increasing $W$ while maintaining a constant $N$ can reduce the spatial separation between the fluid antenna ports; thereby, the spatial correlation decreases and a lower OP is therefore achieved.

\begin{figure}
\centering
\subfigure[]{%
\includegraphics[width=0.25\textwidth]{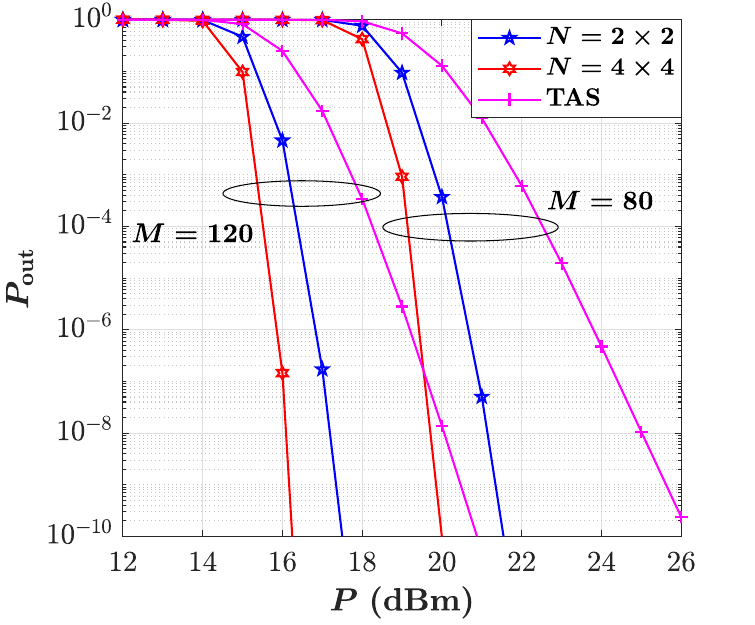}\label{fig-o-p-n}%
}
\subfigure[]{%
\includegraphics[width=0.25\textwidth]{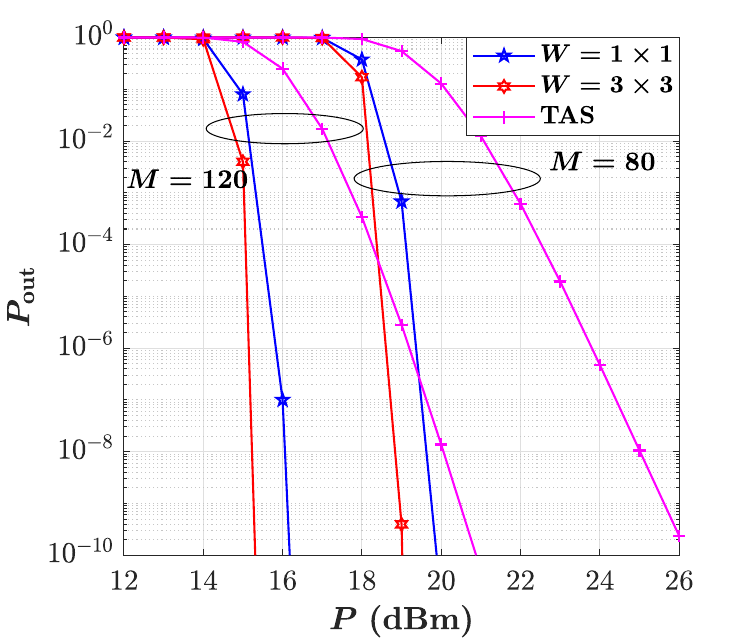}\label{fig-o-p-w}%
}
\vspace{-0.2cm}
\caption{OP versus $P$ (a) for selected values of $M$ and $N$ when $W=1\lambda\times 1\lambda$, and (b) for selected values of $M$ and $W$ when $N=25$.}\label{fig-p}\vspace{-0.5cm}
\end{figure}

To gain more insights into how the number of ports, fluid antenna size, and the number of RIS elements affect system performance, Fig.~\ref{fig-num} is presented. In Fig.~\ref{fig-o-n1}, we can observe that as $N$ steadily increases, the OP initially improves and then becomes saturated to a floor. This is because the spatial correlation between fluid antenna ports increases as $N$ grows under a constant $W$, and thus, a reduction in diversity gain is experienced until it eventually reaches saturation. Moreover, it can be recognized in Fig.~\ref{fig-o-w1} that increasing $W$ under a given large $N$ enhances the OP performance without limitation but at a slow rate for a larger $W$. This behavior is reasonable due to the reduction of spatial correlation by increasing the size of the fluid antenna. However, the interesting point is that such improvements in the OP performance based on $N$ and $W$ variations are achieved if the number of RIS elements $M$ is sufficiently large. For instance, in Fig.~\ref{fig-o-w1}, under a given $M=125$, the OP decreases from the order of $10^{-3}$ to the order of $10^{-7}$ when $W$ raises from $1$ to $9 ~(\lambda^2)$, while under $M=105$, the OP almost remains constant as $W$ increases. Hence, as shown in Figs.~\ref{fig-o-m-n1} and \ref{fig-o-m-w1}, the OP performance for the RIS-aided FAS and RIS-aided TAS is almost similar for a small number of RIS elements, but the impact of FAS on the RIS-aided communication system becomes more remarkable compared with TAS when $M$ continuously increases.


Fig.~\ref{fig-dor} illustrates the behavior of the DOR in terms of the bandwidth $B$ and the amount of data $R$ for different values of $N$, $W$, and $M$. As expected, it is evident in Figs.~\ref{fig-d-b-n1} and \ref{fig-d-b-w1} that a preset data amount can be transmitted with reduced delay as the channel bandwidth increases, where such an achievement becomes more noticeable when the number of RIS elements is larger. Moreover, we can observe that the DOR performance remarkably enhances when $N$ or $W$ grows, which implies that data transmission in the RIS-aided FAS results in lower delays compared to the TAS counterpart. In Figs.~\ref{fig-d-r-n1} and \ref{fig-d-r-w1}, we can also see that the DOR performance weakens by increasing the transmitted data $R$ under a fixed bandwidth $B=2~{\rm MHz}$ because transmission becomes almost impossible at the high date rate; consequently, the latency increases drastically. Nevertheless, it can be observed that by applying FAS instead of TAS in the RIS-aided communication scenario, we can prevent the increase in transmission delay. For instance, sending $R=5~{\rm kbits}$ amounts of data with a low DOR is nearly impossible when TAS is considered at the user, but it can be sent with a small DOR (e.g., in the order of $10^{-8}$ as shown in Fig.~\ref{fig-d-r-n1}) when the FAS is used at the user instead. Therefore, from the ultra reliable low latency communications (URLLC) perspective, FAS with only one active port overtakes TAS in RIS-assisted communication systems.

\begin{figure*}
\centering
\subfigure[]{%
\includegraphics[width=0.25\textwidth]{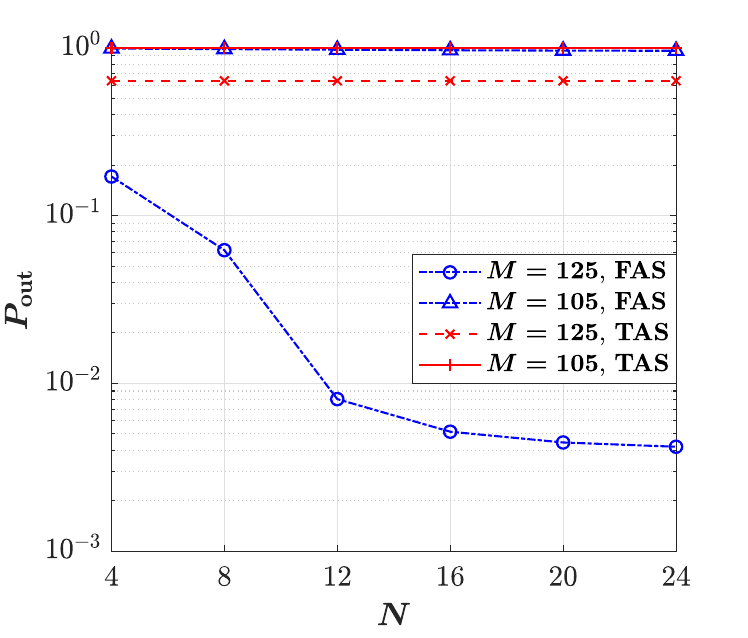}\label{fig-o-n1}%
}
\subfigure[]{%
\includegraphics[width=0.25\textwidth]{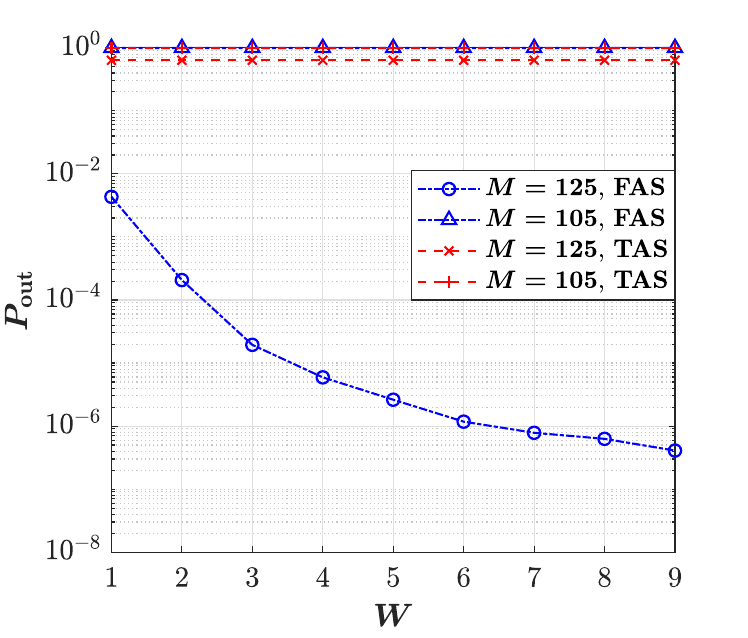}\label{fig-o-w1}%
}
\subfigure[]{%
\includegraphics[width=0.25\textwidth]{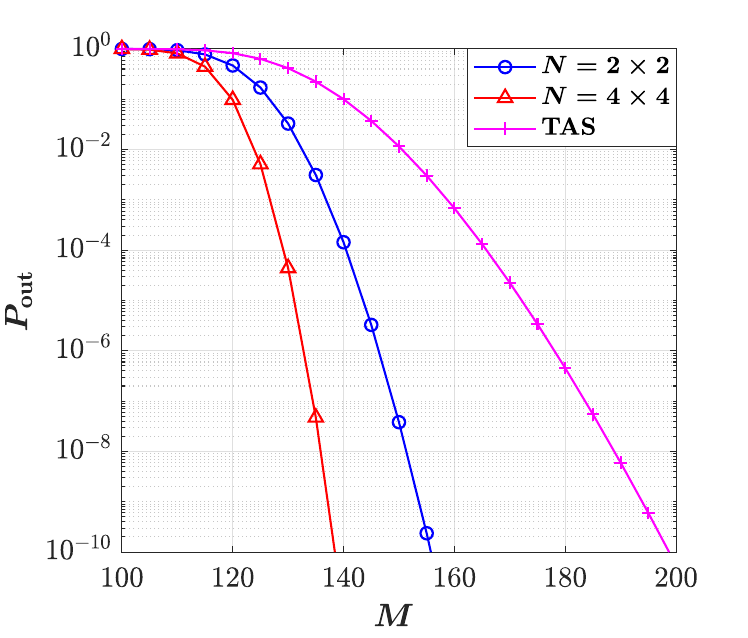}\label{fig-o-m-n1}%
}
\subfigure[]{%
\includegraphics[width=0.25\textwidth]{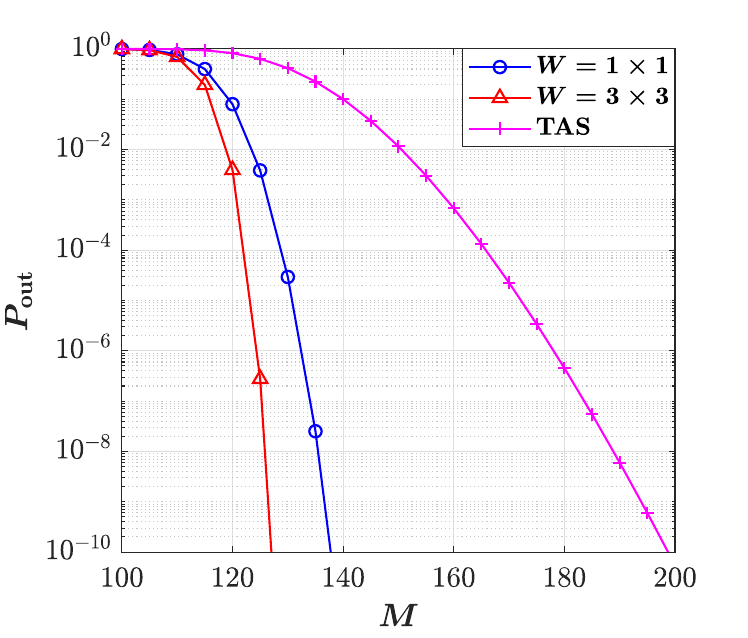}\label{fig-o-m-w1}%
}
\vspace{-0.2cm}
\caption{OP versus (a) the number of fluid antenna ports $N$ for a fixed $W=1\lambda\times 1\lambda$, (b) the size of fluid antenna $W$ for a fixed $N=25$, (c) the number of RIS elements $M$ for a fixed $W=1\lambda\times1\lambda$, and (d) the number of RIS elements $M$ for a fixed $N=25$.}\label{fig-num}
\end{figure*}
\begin{figure*}
\centering
\subfigure[]{%
\includegraphics[width=0.25\textwidth]{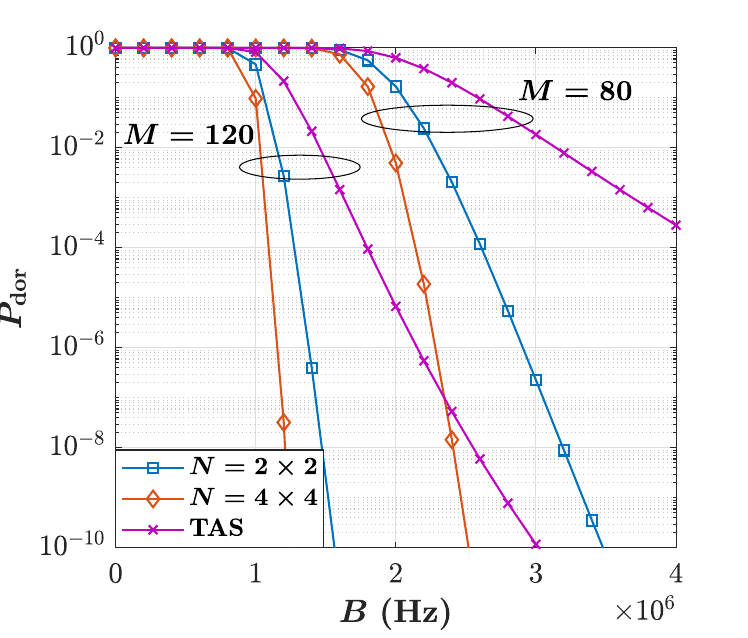}\label{fig-d-b-n1}%
}
\subfigure[]{%
\includegraphics[width=0.25\textwidth]{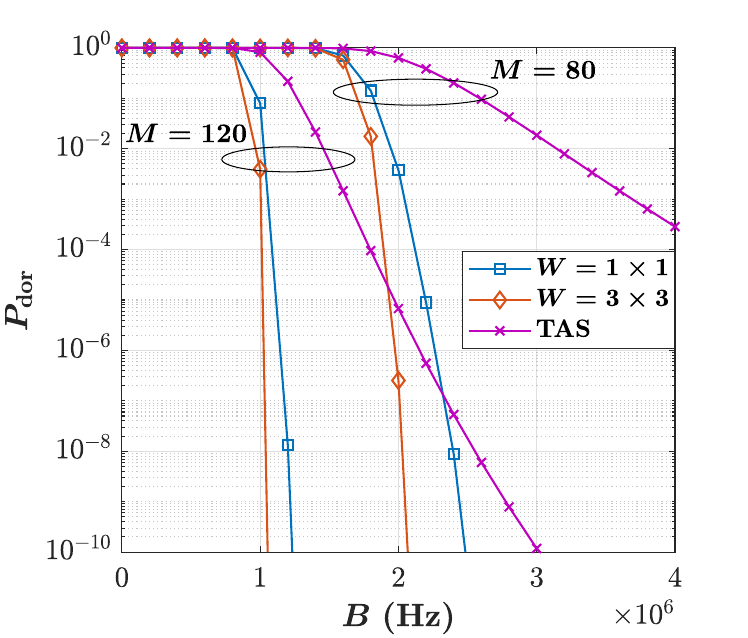}\label{fig-d-b-w1}%
}
\subfigure[]{%
\includegraphics[width=0.25\textwidth]{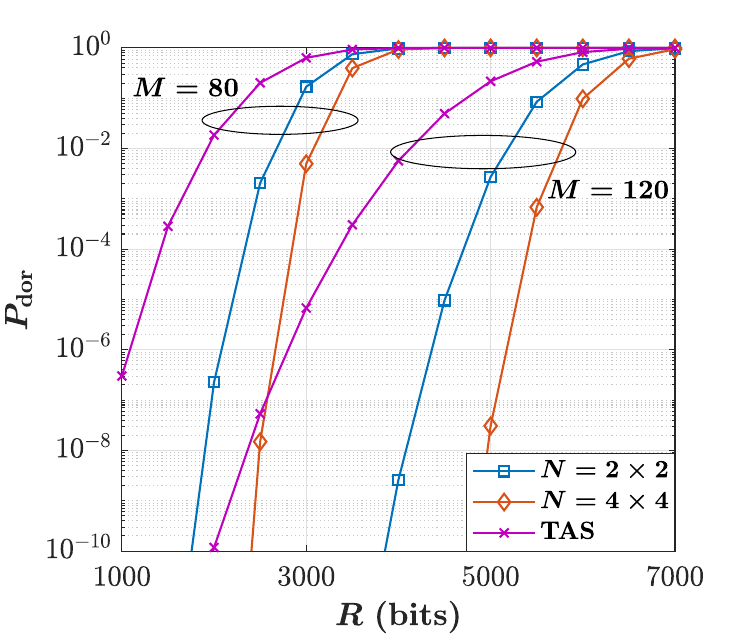}\label{fig-d-r-n1}%
}
\subfigure[]{%
\includegraphics[width=0.25\textwidth]{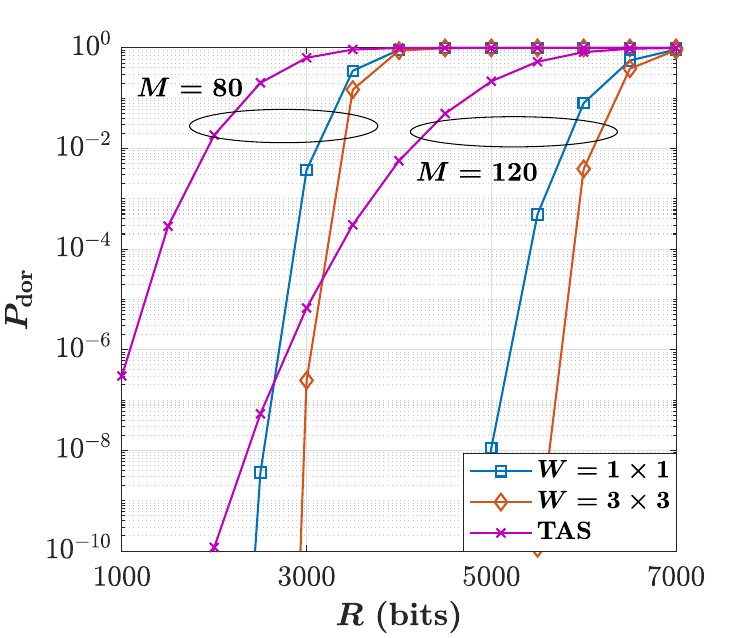}\label{fig-d-r-w1}%
}
\vspace{-0.2cm}
\caption{DOR versus (a) the bandwidth $B$ for a fixed $W=1\lambda\times1\lambda$, (b) the bandwidth $B$ for a fixed $N=25$, (c) the amount of transmitted data $R$ for a fixed $W=1\lambda\times1\lambda$, and (d) the amount of transmitted data $R$ for a fixed $N=25$. }\label{fig-dor}
\end{figure*}

\section{Conclusion}
This letter investigated the performance of RIS-aided FAS, where a fixed single-antenna BS sends information to a FAS-equipped user with the help of an optimized RIS. We first characterized the CDF and PDF of the equivalent channel at the user by modeling the spatial correlation between fluid antenna ports via the Gaussian copula. Then we derived the OP and DOR from compact analytical expressions. The numerical results indicated that using the FAS can provide more reliable and low-latency transmission in RIS-aided communication systems when the number of RIS elements is large.

\appendices
\section{Spatial Correlation in Planar FAS}\label{app1}
Given the defined size and number of ports for the fluid antenna, we consider the position of the $n$-th port as $\vec{n}_n=\left[0, \frac{n_2-1}{N_2-1}W_2, \frac{n_1-1}{N_1-1}W_1\right]^T$. Assuming a plane wave approaches the fluid antenna surface with azimuth angle $\omega$ and elevation angle $\nu$, the array response vector can be formulated as
\begin{align}
	\vec{a}\left(\omega,\nu\right)=\left[\mathrm{e}^{j\vec{k}\left(\omega,\nu\right)^T\vec{n}_1},\dots,\mathrm{e}^{j\vec{k}\left(\omega,\nu\right)^T\vec{n}_{N}}\right]^T,
\end{align}
where $\vec{k}\left(\omega,\nu\right)=\frac{2\pi}{\lambda}\left[\cos\left(\nu\right)\cos\left(\omega\right), \cos\left(\nu\right)\sin\left(\omega\right), \sin\left(\nu\right)\right]^T$ represents the wave vector.

Next, by denoting the normalized spatial correlation matrix $\vec{R}\in\mathbb{C}^{N\times N}$ as $\vec{R}=\mathbb{E}\left\{\vec{a}\left(\omega,\nu\right),\vec{a}\left(\omega,\nu\right)^H\right\}$, the $\left(n,\tilde{n}\right)$ entry of $\vec{R}$ can be expressed as $\varpi_{n,\tilde{n}}=\mathbb{E}\left\{\mathrm{e}^{j2\pi\vec{k}\left(\omega,\nu\right)^T\left(\vec{n}_n-\vec{n}_{\tilde{n}}\right)}\right\}$. For a three-dimensional (3D) isotropic scattering environment over the half-space, we have $f\left(\omega,\nu\right) =\frac{\cos\left(\nu\right)}{2\pi}$ when $\omega\in\left[-\frac{\pi}{2},\frac{\pi}{2}\right]$ amd $\nu\in\left[-\frac{\pi}{2},\frac{\pi}{2}\right]$ \cite{bjornson2020rayleigh}. Hence, utilizing Euler's formula, $\varpi_{n,\tilde{n}}$ can be determined as\vspace{-0.1cm}
\begin{align}
\varpi_{n,\tilde{n}}&=\int_{-\pi/2}^{\pi/2}\int_{-\pi/2}^{\pi/2}\mathrm{e}^{j\frac{2\pi}{\lambda}\Vert\vec{n}_n-\vec{n}_{\tilde{n}}\Vert\sin\left(\nu\right)}f\left(\omega,\nu\right)\mathrm{d}\nu\mathrm{d}\omega\\
&=\frac{\sin\left(\frac{2\pi}{\lambda}\Vert\vec{n}_n-\vec{n}_{\tilde{n}}\Vert\right)}{\frac{2\pi}{\lambda}\Vert\vec{n}_n-\vec{n}_{\tilde{n}}\Vert},\label{eq-app1}
\end{align}
where \eqref{eq-app1} is equal to \eqref{eq-cov} and the proof is accomplished.  \vspace{-0.2cm}
\bibliographystyle{IEEEtran}

\end{document}